\documentclass[twocolumn,aps,prd,amsmath,amsfonts,nofootinbib]{revtex4-1}

\newcommand{\ud}{\mathrm{d}}

\newcommand{\bs}{\begin{split}}
\newcommand{\es}{\end{split}}

\usepackage{xcolor}
\usepackage{verbatim}
\usepackage{float}
\usepackage{color}
\usepackage[hyperindex,breaklinks]{hyperref}
\hypersetup{
colorlinks,
    linkcolor={black},
    citecolor={blue!50!black},
    urlcolor={blue!80!black},
    pdftitle={Quantum effective potential from discarded degrees of freedom},
    pdfauthor={Luke M. Butcher}}
\usepackage{slashed}
\usepackage{graphicx}
\usepackage{subfigure}
\usepackage{psfrag}
\usepackage{mathtools}
\usepackage{amsthm}
\usepackage{amsfonts}
\theoremstyle{plain}

\newtheorem*{Claim}{Claim}

\begin{document}

\title{Quantum effective potential from discarded degrees of freedom}
\author{Luke M. Butcher}
\email[]{lmb@roe.ac.uk}
\affiliation{Institute for Astronomy, University of Edinburgh, Royal Observatory, Edinburgh EH9 3HJ, United Kingdom}
\date{July 10, 2018}
\pacs{}

\begin{abstract}
I obtain the quantum correction $\Delta V_\mathrm{eff}= (\hbar^2/8m) [(1- 4\xi \frac{d+1}{d})(\mathcal{S}')^2 + 2(1-4\xi)\mathcal{S}'']$ that appears in the effective potential whenever a compact $d$-dimensional subspace (of volume $\propto \exp[\mathcal{S}(x)]$) is discarded from the configuration space of a nonrelativistic particle of mass $m$ and curvature coupling parameter $\xi$. This correction gives rise to a force $-\langle\Delta V_\mathrm{eff}'\rangle$ that pushes the expectation value $\langle x\rangle$ off its classical trajectory. Because $\Delta V_\mathrm{eff}$ does not depend on the details of the discarded subspace, these results constitute a generic model of the quantum effect of discarded variables with maximum entropy/information capacity $\mathcal{S}(x)$.
\end{abstract}

\maketitle

\section{Introduction}
It is often possible and desirable to ignore specific degrees of freedom of a system, and focus on those that remain. For example, consider a nonrelativistic  particle in a curved two-dimensional space
\begin{align}\label{tubemetric}
\ud s^2 &= \ud x^2 + [b(x)]^2\ud \phi^2, & &(x,\phi)\in \mathbb{R}\times [0,2\pi),
\end{align}
as illustrated in figure \ref{tube}. If the particle also encounters a potential $V_0(x)$ then its action is
\begin{align}\label{tubeS}
\mathcal{I}[x(t),\phi(t)]=\int \ud t \left[\frac{m}{2} \left(\dot{x}^2 +b^2 \dot{\phi}^2\right) -V_0\right],
\end{align}
giving rise to the following equations of motion:
\begin{align}\label{clasx}
m \ddot{x}&= m b'b \dot{\phi}^2 - V_0',
\\ \label{clasphi}
m b^2 \dot{\phi}&=p_\phi =\text{const.}
\end{align}
Now suppose we only wish to describe the behaviour of the $x$ coordinate of this particle -- perhaps $\phi$ is unobservable in practice, or happens to be irrelevant to whatever applications we have in mind. At the classical level, we can separate the $x$-motion from the $\phi$-motion as follows. Let us write the action (\ref{tubeS}) as 
\begin{align}\nonumber
\mathcal{I}[x(t),\phi(t)]&=\int \ud t \left[\frac{m}{2} \dot{x}^2  +\frac{1}{2mb^2}\left(mb^2 \dot{\phi} - p_\phi\right)^2 \right.\\\label{Asplit} 
&\quad\left. {} + \dot{\phi}p_\phi - \frac{p_\phi^2}{2mb^2} - V_0  \right],
\end{align}
and note that
\begin{align}\nonumber
&\frac{\delta}{\delta x(t)}\int \ud t \left[\frac{1}{2mb^2}\left(mb^2 \dot{\phi} - p_\phi\right)^2 +\dot{\phi}p_\phi \right]
\\\label{Asplit2} & =-\frac{b'}{mb^3}\left(mb^2 \dot{\phi} - p_\phi\right)^2 +  \frac{2b' \dot{\phi}}{b}\left(mb^2 \dot{\phi} - p_\phi\right),
\end{align}
which vanishes on the $\delta/\delta\phi$ equation of motion (\ref{clasphi}). If we only want to determine $x(t)$, we can therefore discard the second and third terms in (\ref{Asplit}) and work with the reduced action:
\begin{align}\label{reducedtubeS}
\mathcal{I}[x(t)]\equiv\int \ud t \left[\frac{m}{2}\dot{x}^2 -V_\mathrm{cl}\right],
\end{align}
where the (classical) effective potential is
\begin{align}\label{Veffclas}
V_\mathrm{cl}=V_0 + \frac{p_\phi^2}{2m b^2}.
\end{align}
The reduced action (\ref{reducedtubeS}) generates the correct equations of motion for $x$, consistent with substituting (\ref{clasphi}) into (\ref{clasx}), and allows us to treat the particle as though it were living in a reduced configuration space
\begin{align}\label{redmetric}
\ud s^2 &= \ud x^2,  & &x\in \mathbb{R}.
\end{align}
We no longer need to refer to $\phi$, and can think of $p_\phi$ as a  parameter of the system. For a concrete application of this formalism, recall Newtonian orbital mechanics: with $b(x)=x$, the metric (\ref{tubemetric}) describes a flat plane with radial coordinate $x$, and $V_\mathrm{cl}= V_0+p_\phi^2/2m x^2$ is the standard centrifugal potential. 

\begin{figure}[t]
\centering
\includegraphics[scale=.5]{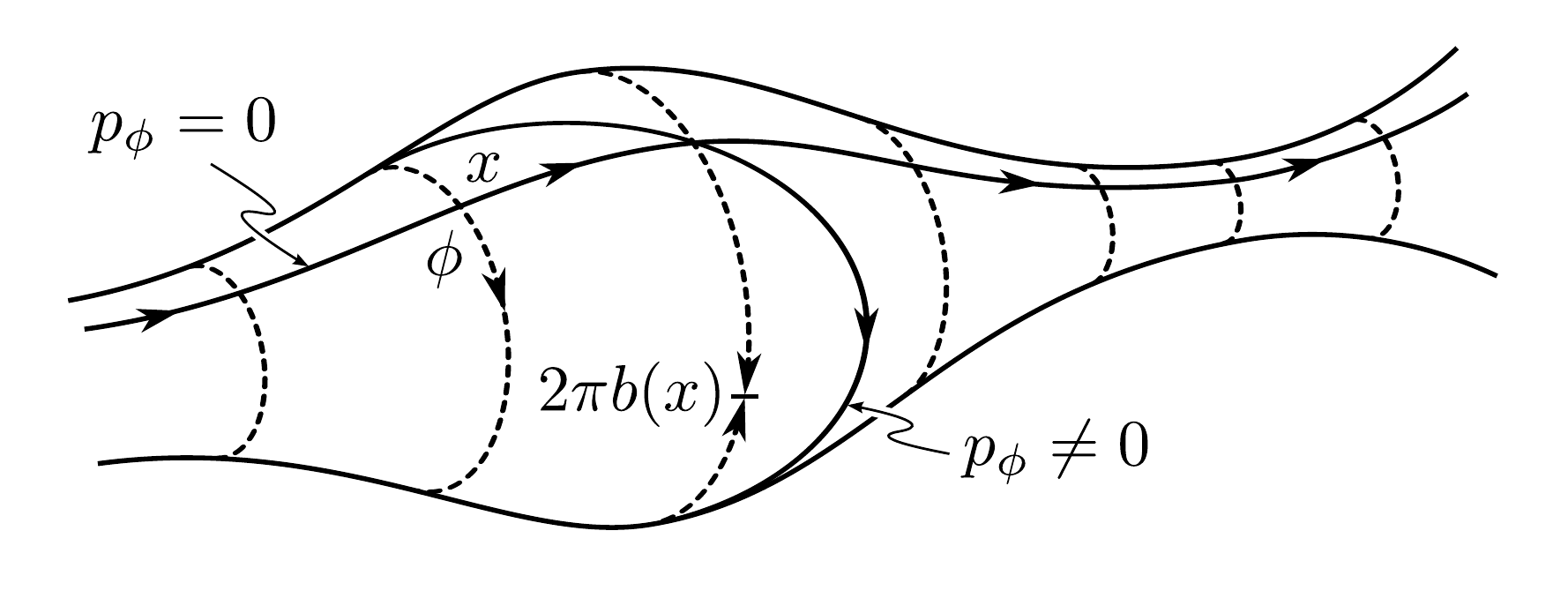}\vspace{-.5cm}
\caption{As classical particles move over the curved space (\ref{tubemetric}) their $x$ coordinate can be predicted without reference to $\phi$, using the reduced action (\ref{reducedtubeS}). However, quantum particles encounter an addition potential $\Delta V_\mathrm{eff}$ due to variations in the physical size $2\pi b$ of the discarded subspace $\phi\in [0,2\pi)$.}\label{tube}
\end{figure}

It is important to realise, however, that once quantum effects are considered, the above procedure is no longer valid. If we na\"ively quantize the reduced system (\ref{reducedtubeS}) we will not arrive at the correct result: that obtained by quantizing the original system (\ref{tubeS}) and \emph{then} reducing its configuration space. As we will see, the correct result differs from the na\"ive one by a quantum correction to the effective potential $\Delta V_\mathrm{eff}$, dependent on the physical size $\mathrm{Vol}_\phi=2\pi b(x)$ of the discarded subspace $\mathcal{M}_\phi\cong [0,2\pi)$. 

\section{Discarding a Single Variable}
Let us work in the Schr\"odinger picture, and first confirm the existence of $\Delta V_\mathrm{eff}$ for the simple system above. As usual, we describe the quantum particle with a wavefunction, a scalar field $\Psi(x,\phi,t)$ that defines coordinate-invariant probabilities via integrals of the form
\begin{align}\label{prob}
P = \int \ud x \,\ud\phi \sqrt{g}|\Psi|^2,
\end{align}
where $\sqrt{g}\equiv\sqrt{\det(g_{ij})}$ is the covariant measure endowed by the metric $g_{ij}$. [(\ref{tubemetric}) $\Rightarrow \sqrt{g}=b$.] In curved space, $\Psi$ obeys the covariant Schr\"odinger equation\footnote{The first systematic treatment of quantum mechanics in curved space is due to DeWitt \cite{DeWitt57} whose paper includes a canonical derivation of the covariant Schr\"odinger equation. For another perspective on the origin and ambiguity of the curvature term, see \cite{DeWitt-Morette80}. A more modern approach can be found in \cite{BLASZAK13}.}:
\begin{align}\label{covSchro}
i \hbar \partial_t \Psi =\left[\frac{\hbar^2}{2m}\left(-\nabla^2 + \xi R\right) + V_0\right]\Psi,
\end{align}
where the Laplacian
\begin{align}
\nabla^2 = \frac{1}{\sqrt{g}}\partial_i \sqrt{g}g^{ij}\partial_j
\end{align}
and the Ricci scalar $R\equiv R_{ij}g^{ij}\equiv R^{k}{}_{ikj}g^{ij}$ are constructed from the metric $g_{ij}$.\footnote{The curved tube (\ref{tubemetric}) is the \emph{entire} configuration space of the system, so the covariant Schr\"odinger equation (\ref{covSchro}) can refer only to the \emph{intrinsic} geometry of this manifold. Contrast this with a particle that actually exists in $\mathbb{R}^3$, but is constrained to a two-dimensional surface $\Sigma \subset \mathbb{R}^3$ by a steep potential well: here, the extrinsic curvature of $\Sigma$ will also play a role \cite{Jensen71,daCosta81,SILVA15,Ortix15,Wang17}. } [(\ref{tubemetric}) $\Rightarrow R=-2 b''/b$.] The form of (\ref{covSchro}) is fixed by coordinate invariance, unitarity, locality, dimensional considerations, and the limits $R\to0$, $V_0\to0$; however, the curvature coupling parameter $\xi\in \mathbb{R}$ is an arbitrary dimensionless constant, representing a quantization ambiguity of the system \cite{DeWitt57,DeWitt-Morette80,BLASZAK13}. One can choose to invoke `minimal coupling' $\xi=0$, or motivate a conformal coupling according to some other principle or consideration \cite{CALLAN70,Karamatskou14}. For the sake of generality, we leave $\xi$ unspecified.

Having quantized the original system, we proceed to discard the $\phi$ subspace. In order to make $p_\phi$ a parameter of the system, we must first insist that the particle be in an eigenstate of the angular momentum operator:
\begin{align}\label{eigenstates}
  \hat{p}_\phi \Psi\equiv - i \hbar \partial_\phi \Psi=  p_\phi \Psi.
\end{align} 
This requirement serves as the analogue of equation (\ref{clasphi}) and ensures that the particle's $\phi$ behaviour is sufficiently simple that the $x$ dynamics can be described in isolation. The states of interest are then
\begin{align}\label{Psix}
\Psi &= \frac{e^{i k\phi}}{\sqrt{2 \pi b(x)}} \Psi_x(x,t) , &p_\phi/\hbar= k&\in\mathbb{Z},
\end{align}
where the normalisation of $\Psi_x$ ensures that probabilities (\ref{prob}) become integrals of the form
\begin{align}\label{probx}
P = \int \ud x|\Psi_x|^2,
\end{align}
without any reference to the $\phi$ subspace. Hence we can think of $\Psi_x$  as the wavefunction of the particle on the reduced configuration space (\ref{redmetric}).

To obtain the evolution equation for $\Psi_x$, we simply insert (\ref{Psix}) into (\ref{covSchro}). We arrive at a reduced Schr\"odinger equation
\begin{align}\label{reducedSchro}
i \hbar \partial_t \Psi_x = \left[-\frac{\hbar^2}{2m}\partial_x^2 + V_\mathrm{qu}\right]\Psi_x,
\end{align}
where the \emph{quantum} effective potential  
\begin{align}\label{Veffquant}
V_\mathrm{qu}&\equiv V_\mathrm{cl} + \Delta V_\mathrm{eff}
\end{align}
has an additional contribution
\begin{align}\label{DeltaV}
\Delta V_\mathrm{eff} = \frac{\hbar^2}{2m} \left[-\frac{1}{4}\left(\frac{b'}{b}\right)^2 + \frac{1-4\xi}{2}\left(\frac{b''}{b}\right)\right],
\end{align}
as promised. There are a few things to note about this quantum correction. First, this effect is not purely a consequence of spatial curvature: even for the flat case $b=x$ we find $\Delta V_\mathrm{eff}\ne0$. Second, $\Delta V_\mathrm{eff}$ cannot be made to vanish identically by some choice of $\xi$. Third, $\Delta V_\mathrm{eff}$ does not depend on $p_\phi$, so all states (\ref{Psix}) experience the same correction.

To illustrate an important physical implication of $\Delta V_\mathrm{eff}$, we now consider an arbitrary state:
\begin{align}\label{arb2d}
\Psi &= \sum_{k=-\infty}^{\infty} \frac{e^{i k\phi}}{\sqrt{2 \pi b(x)}} \Psi^{k}_x(x,t),
\end{align}
where $ \int \ud x \sum _k |\Psi_x^k|^2=1$ ensures proper normalisation. The expectation value of a function $f(x,p_\phi)$ is then
\begin{align}\nonumber
\left\langle f(x,p_\phi)\right\rangle&\equiv \int \ud x\, \ud \phi \sqrt{g} \Psi^* f(x,-i\hbar \partial_\phi)\Psi\\\label{expdef}
&=  \int \ud x \sum_{k=-\infty}^{\infty} | \Psi_x^{k}|^2 f(x,\hbar k).
\end{align}
As each $\Psi_x^k$ obeys the reduced Schr\"odinger equation (\ref{reducedSchro}) with $p_\phi=\hbar k$, it follows that $\langle x\rangle$ evolves according to 
\begin{align}\label{meanEOM}
m\partial_t^2\langle x\rangle&=- \langle V'_\mathrm{cl} + \Delta V_\mathrm{eff}'\rangle .
\end{align}
We see that the classical equation of motion $m\ddot{x}=-V'_\mathrm{cl}$ no longer holds true on average. Indeed, the average deviation from the classical equation is given by 
\begin{align}\nonumber
m\partial_t^2\langle x\rangle+ \left\langle V'_\mathrm{cl}\right\rangle&= -\langle\Delta V_\mathrm{eff}'\rangle\\
&=- \int \ud x\Delta V'_\mathrm{eff}\sum_{k=-\infty}^\infty |\Psi_x^k|^2,
\end{align}
which only depends on $\Delta V'_\mathrm{eff}$ and the probability density over the reduced configuration space:
\begin{align}
\rho_x(x,t)\equiv \int \ud \phi \sqrt{g} |\Psi |^2= \sum_{k=-\infty}^\infty |\Psi_x^k|^2.
\end{align}
Hence $\Delta V_\mathrm{eff}$ directly influences the average motion of the particle, independent of the internal details of the quantum state.

Had we na\"ively quantized the reduced system (\ref{reducedtubeS}) we would not have included the quantum contribution (\ref{DeltaV}) to our effective potential. Moreover, nothing about the reduced action (\ref{reducedtubeS}) nor the configuration space (\ref{redmetric}) would have hinted at the error we were making -- we would encounter no striking technical difficulties or operator ambiguities. As such, this result serves as a general warning to those attempting to quantize \emph{any system} in which degrees of freedom have already been discarded: if the physical size of the discarded configuration space varies as a function of the remaining variables, then one expects to miss an effective potential similar to (\ref{DeltaV}). For instance, the mini-superspace approach to quantum cosmology \cite{Kiefer} will need to quantify the volume of configuration space neglected in assuming a highly symmetric universe.  

\section{Discarding a Generic Subspace}\label{DiscGen}
Thus far, we have obtained the quantum correction (\ref{DeltaV}) that arises from the removal of a one-dimensional subspace. To generalise this result, let us now consider a nonrelativistic quantum particle in $D=d+1$ spatial dimensions:
\begin{align}\label{genmetric}
\ud s^2 &= \ud x^2 + [b(x)]^2 \tilde{g}_{IJ}(\phi)\ud \phi^I \ud \phi^J,
\end{align}
where $\tilde{g}_{IJ}$ is the metric of a compact $d$-dimensional manifold $\mathcal{M}_\phi$ with coordinates $\phi\equiv(\phi^1,\ldots, \phi^d)$. As before, we seek a description of the dynamics in which we can ignore $\phi$ and treat the particle as though it were living in the reduced configuration space (\ref{redmetric}) with an effective potential that depends on a single parameter. In fact (\ref{genmetric}) is the most general metric that allows for this type of reduction: see the appendix for a proof.

At the classical level, the analysis follows steps (\ref{tubeS})--(\ref{Veffclas}) with minor modifications. In place of $p_\phi$, we assemble 
\begin{align}\label{clasE}
E_\phi\equiv\frac{1}{2m} \tilde{g}^{IJ}p_{\phi^I} p_{\phi^J}=\frac{mb^4}{2} \tilde{g}_{IJ}\dot{\phi}^I\dot{\phi}^J,
\end{align}
which is conserved by virtue of the $\delta/\delta\phi^I$ equations of motion. Writing $ |\dot{\phi}|\equiv (\tilde{g}_{IJ}\dot{\phi}^I\dot{\phi}^J)^{1/2}$, we therefore have
\begin{align}
mb^2 |\dot{\phi}|=\sqrt{2 m E_\phi} = \mathrm{const},
\end{align}
as a substitute for equation (\ref{clasphi}). Hence steps (\ref{Asplit}) and (\ref{Asplit2}) now follow with replacements $\dot{\phi}\to |\dot{\phi}|$, $p_\phi\to\sqrt{2mE_\phi}$. This generates the reduced action (\ref{reducedtubeS}) with
\begin{align}\label{genclasV}
V_\mathrm{cl}=V_0 + E_\phi/b^2
\end{align} 
as the classical effective potential.

On the quantum side, we begin by observing that the Laplacian and Ricci scalar can be decomposed as follows:
\begin{align} \label{decomp1}
\nabla^2&= b^{-d}\partial_x b^d\partial_x + b^{-2}\tilde{\nabla}^2,\\ \label{decomp2}
R&=d(1-d) (b'/b)^2 - 2 d (b''/b) + \tilde{R}/ b^{2},
\end{align}
where $\tilde{\nabla}^2$ and $\tilde{R}$ are constructed from the metric $\tilde{g}_{IJ}$. The states (\ref{Psix}) generalise to
\begin{align}\label{PsixD}
\Psi = \frac{\Phi(\phi)}{[b(x)]^{d/2}}  \Psi_x(x,t),
\end{align}
where $\Phi$ is an `energy' eigenfunction on $\mathcal{M}_\phi$,
\begin{align}\label{eigenD}
\frac{\hbar^2}{2m}\left(-\tilde{\nabla}^2  + \xi \tilde{R}\right)\Phi=E_\phi\Phi,
\end{align}
with unit norm:
\begin{align}\label{Nphi}
\int \ud^d\phi \sqrt{\tilde{g}} |\Phi|^2=1.
\end{align} 
It is natural to identify $E_\phi$, defined in (\ref{clasE}), with the eigenvalue of (\ref{eigenD}) because the covariant Schr\"odinger equation (\ref{covSchro}) follows the same quantisation rule $g^{ij} p_i p_j \to \hbar^2[-\nabla^2 +\xi R]$. As before, we have normalised $\Psi_x$ such that probabilities have the standard form (\ref{probx}).

Inserting (\ref{decomp1})--(\ref{eigenD}) into the Schr\"odinger equation (\ref{covSchro}) we obtain the reduced Schr\"odinger equation (\ref{reducedSchro}) once again. The quantum effective potential (\ref{Veffquant}) now differs from its classical counterpart (\ref{genclasV}) by
\begin{align}\nonumber
\Delta V_\mathrm{eff} &= \frac{\hbar^2 d}{2m} \left[ \left(\frac{d-2}{4} + \xi (1-d)\right)\left(\frac{b'}{b}\right)^2  \right.
\\\label{DeltaVD}
&\quad {} \left.+\frac{1-4\xi}{2}\left(\frac{b''}{b}\right)\right].
\end{align}
We see that $\Delta V_\mathrm{eff}$ is independent of $\tilde{g}_{IJ}$ and $E_\phi$, and does not vanish identically for any $(\xi,d)\in \mathbb{R}\times \mathbb{N}$.

It is easy to check that  $\Delta V_\mathrm{eff}$ appears in the average equation of motion (\ref{meanEOM}) just as before. The only change to this calculation is that the arbitrary state (\ref{arb2d}) is now
\begin{align}
\Psi &= \sum_{k} \frac{\Phi_k(\phi)}{ [b(x)]^{d/2}}  \Psi^k_x(x,t),
\end{align}
where the $\{\Phi_k\}$ are eigenfunctions (\ref{eigenD}) with eigenvalues $\{E_\phi^k\}$, forming an orthonormal basis over $\mathcal{M}_\phi$:
\begin{align}\label{orth}
\int \ud^d\phi \sqrt{\tilde{g}} \, \Phi^*_{k}\Phi_{k'} =\delta_{kk'}.
\end{align} 
Consequently, we compute expectation values with
\begin{align}\nonumber
\left\langle f(x,E_\phi)\right\rangle&\equiv \int \ud x\,\ud^d \phi \sqrt{g} \Psi^* f\left(x,\frac{\hbar^2}{2m}\left[-\tilde{\nabla}^2  + \xi \tilde{R}\right]\right)\Psi\\
&=  \int \ud x \sum_{k}| \Psi_x^{k}|^2 f(x,E^k_\phi),
\end{align}
instead of equation (\ref{expdef}).

\section{Discarded Information Capacity}
The power of equation (\ref{DeltaVD}) is revealed by expressing this result in terms of the \emph{information} we discard by ignoring the degrees of freedom in $\phi$. To quantify this information, we first need to regularise the infinite-dimensional Hilbert space of the particle. Let us imagine dividing the curved space (\ref{genmetric}) into a lattice of small cells with spacing $\ell\ll\min\{ b,(b/b'),\sqrt{b/b''}\}$. Then, at a given value of $x$, the particle can be in any of 
\begin{align}
\Omega(x)=  \frac{\mathrm{Vol}_\phi(x)}{\ell^d}\propto \frac{[b(x)]^d}{\ell^d}
\end{align}
locations on the $\phi$ sub-lattice, and the maximum entropy (or information) that can be stored in the $\phi$ subspace is given by Boltzmann's formula:
\begin{align}
\mathcal{S}&= \ln \Omega.
\end{align}
With this in mind, equation (\ref{DeltaVD}) can be written as
\begin{align}\label{VofS}
\!\Delta V_\mathrm{eff}&= \frac{\hbar^2}{8m}\left[\!\left(1- 4\xi \frac{d+1}{d}\right)\!(\mathcal{S}')^2 + 2(1-4\xi)\mathcal{S}''\right]\!.
\end{align}
Crucially, this formula is completely independent of the arbitrary length $\ell$. If we wish, we can now take $\ell\to 0$ and safely return to the continuum limit. Thus we have obtained a robust relation between the quantum correction and the maximum entropy/information capacity $\mathcal{S}$ of the discarded subspace.

\section{Semiclassical Action}
Because $\Delta V_\mathrm{eff}$ does not depend on the details of the discarded space $\mathcal{M}_\phi$, equation (\ref{VofS}) can be used to model discarded variables \emph{in general}. To illustrate this idea, suppose we are interested in predicting the behaviour of an observable $x$ in a system which is not well-understood, but is known to have the following two properties. First, the classical motion $x(t)$ can be derived from an action (\ref{reducedtubeS}) without reference to other dynamical variables. Second, for a given value of $x$, the system can store exactly $\mathcal{S}(x)$ nats of information. Now, even if we know nothing about the `discardable' degrees of freedom that hold the information, we can still model their effect on the quantum behaviour of $x$. The method is simple: treat the full configuration space as (\ref{genmetric}) and leave the internal metric $\tilde{g}_{IJ}$ unspecified. We conclude that the information-holding variables introduce a quantum correction (\ref{VofS}) to the effective potential, regardless of the details of their configuration space. Furthermore, we can account for this effect \emph{semiclassically} by replacing the classical action (\ref{reducedtubeS}) with
\begin{align}\label{semiclasS}
\mathcal{J}[x(t)]=\int \ud t \left[\frac{m}{2}\dot{x}^2 -\left(V_\mathrm{cl} + \Delta V_\mathrm{eff}\right)\right],
\end{align}
which generates equations that accurately capture the motion (\ref{meanEOM}) of the expectation value $\langle x\rangle$. Provided $\mathcal{S}(x)$ is known, the semiclassical action (\ref{semiclasS}) only introduces two parameters $(\xi,d)$ that would need to be determined experimentally. Thus equations (\ref{VofS}) and (\ref{semiclasS}) constitute a powerful semiclassical model of the quantum effect of discarded variables.

The semiclassical action (\ref{semiclasS}) also lets us express the propagator for $\Psi_x$ as a path integral over the reduced configuration space:
\begin{align}\label{prop}
K(x_f,x_0;T)=\int_{x(0)=x_0}^{x(T)=x_f} \mathcal{D}x(t)\, e^{i \mathcal{J}[x(t)]/\hbar}.
\end{align}
This relation is evident from the reduced Schr\"odinger equation (\ref{reducedSchro}) and the standard path integral construction \cite{Feyn48,Kleinert}.  Furthermore, it must also be possible to derive the reduced propagator (\ref{prop}) from the path integral over the \emph{entire} configuration space $\int \mathcal{D}x(t)\,\mathcal{D}^d\phi(t)\ldots$, by `integrating out' the paths in $\phi$.\footnote{There is some ambiguity in the definition of this path integral, each resolution of which fixes the value of $\xi$ in the effective potential (\ref{VofS}). Storchak has performed a similar treatment of the \emph{phase space} path integral with $\xi=0$ \cite{Storchak93}.} In a future publication, I will demonstrate this process explicitly, and hence provide a derivation of $\Delta V_\mathrm{eff}$ and $\mathcal{J}[x(t)]$ that does not require the Schr\"odinger equation.

\section{Conclusions}
When predicting the classical motion of an observable $x$, other degrees of freedom can often be ignored and subsumed into an effective potential. However, if these discarded variables have a configuration space which varies in size as a function of $x$, so that the discarded information capacity is $\mathcal{S}(x)$, then quantum effects generate an additional term in the effective potential (\ref{VofS}). This quantum correction directly influences the observable behaviour, forcing the expectation value $\langle x\rangle$ away from its classical trajectory (\ref{meanEOM}). The semiclassical action (\ref{semiclasS}) accounts for this phenomenon within the equation of motion, and also generates the path integral propagator (\ref{prop}) over the reduced configuration space (\ref{redmetric}). In general, the quantum correction is determined by $\mathcal{S}(x)$, the number of discarded variables $d\in \mathbb{N}$, and the curvature coupling parameter $\xi\in \mathbb{R}$; beyond this, the details of the discarded configuration space are irrelevant. As such, these results constitute a powerful general-purpose model for the quantum effect of discarded variables, applicable to any observable whose classical motion is determined by the standard nonrelativistic action (\ref{reducedtubeS}).

\section*{Acknowledgements}
The author is supported by a research fellowship from the Royal Commission for the Exhibition of 1851. He also thanks the anonymous reviewer whose comments motivated the appendix.

\appendix

\section{Reducible Spaces}
We have seen that the $x$-coordinate of a particle can sometimes be predicted without detailed knowledge of the other degrees of freedom: a \emph{single parameter} (e.g.\ $p_\phi$ or $E_\phi$) contains all the information we need. Here, I prove that the only manifolds that allow this kind of reduction are those with a metric of the form (\ref{genmetric}).

Let us begin with an arbitrary metric in $D=d+1$ spatial dimension:
\begin{align}\label{genmetric2}
\ud s^2 = g_{ij}(q)\ud q^i \ud q^j,
\end{align}
where the coordinates $q\equiv(q^1,\ldots,q^{d+1})\equiv(x,\phi^1,\ldots,\phi^d)$ cover an open region $\mathcal{U}\cong(x_-,x_+)\times\mathcal{U}_\phi$, with $\mathcal{U}_\phi$ homeomorphic to an open $d$-ball. A particle of mass $m$ moving in $\mathcal{U}$, under the influence of a potential $V_0$, will have a Hamiltonian
\begin{align}\label{geoH}
H= \frac{1}{2m}g^{ij}(q)p_i p_j + V_0(q),
\end{align}
where $g^{ij}$ is the inverse of $g_{ij}$, and $\{p_i\}$ are the momenta conjugate to $\{q^i\}$. 

We aim to discard the coordinates $\phi \equiv (\phi^1,\ldots,\phi^d)$ and momenta $p_\phi\equiv (p_{\phi^1},\ldots,p_{\phi^d})$ and predict the motion of $(x,p_x)$ from a one-dimensional Hamiltonian
\begin{align}\label{redH}
H = \frac{p_x^2}{2 \tilde{m}(\lambda)}  + V_\mathrm{cl}(x,\lambda),
\end{align}
where the effective mass $\tilde{m}$ and effective potential $V_\mathrm{cl}$ can depend on a single real parameter $\lambda=\lambda(x,p_x,\phi,p_\phi)\in \mathbb{R}$. (We take $\tilde{m}$, $V_\mathrm{cl}$, and $\lambda$ to be $C^1$ functions.) In order that $\lambda$ can be treated as a \emph{parameter of the system}, it must be conserved along every trajectory, and must not spoil the canonical equations of motion:
\begin{align}\label{refHeom}
\dot{x}&=\frac{\partial H}{\partial p_x}= \frac{p_x}{\tilde{m}},&
\dot{p}_x&=-\frac{\partial H}{\partial x}= -\partial_x V_\mathrm{cl}.
\end{align}
But note that (\ref{redH}) will generate (\ref{refHeom}) if and only if
\begin{align}\label{consistent}
 \frac{\partial H}{\partial \lambda} \frac{\partial \lambda}{\partial x} =\frac{\partial H}{\partial \lambda} \frac{\partial \lambda}{\partial p_x}   =0
\end{align}
everywhere. Thus, $H$ can depend on $\lambda$ only where $\lambda$ is independent of $(x,p_x)$. It is therefore safe to assume that $\lambda=\lambda(\phi,p_\phi)$, independent of $(x,p_x)$, without any loss in the generality of $H$. Consequently, $\lambda:T^*\mathcal{U}_\phi \to \mathbb{R}$ maps the cotangent bundle $T^*\mathcal{U}_\phi\cong\mathcal{U}_\phi\times \mathbb{R}^d$ onto
\begin{align}\label{lambdadef}
\Lambda \equiv \lambda\left(T^*\mathcal{U}_\phi\right)\subseteq \mathbb{R},
\end{align}
which represents the range of possible values for $\lambda(\phi,p_\phi)$.

Our task is to determine which geometries (\ref{genmetric2}) have Hamiltonians (\ref{geoH}) consistent with the reduced form (\ref{redH}). In other words, we require
\begin{align}\label{compareH1}
\frac{1}{2m}g^{ij}(q)p_i p_j + V_0(q)=\frac{p_x^2}{2 \tilde{m}(\lambda)}  + V_\mathrm{cl}(x,\lambda),
\end{align}
for all $(x,\phi,p_x,p_\phi)\in T^*\mathcal{U}$. Recalling $\lambda=\lambda(\phi,p_\phi)$, we can compare powers of $p_x$ in (\ref{compareH1}) and immediately extract
\begin{align}
g^{xx}(x,\phi)&= \frac{m}{\tilde{m}(\lambda(\phi,p_\phi))},& g^{x\phi}&=0.
\end{align}
But given that $g^{xx}$ does not depend on $p_\phi$, we must either have $\lambda=\lambda(\phi)$ or $\tilde{m}=\mathrm{const}$. Even in the former case we still have $\tilde{m}=\tilde{m}(\lambda(\phi)) =\mathrm{const}$, because $\lambda= \lambda(\phi)$ cannot be conserved along every trajectory unless it is actually independent of $\phi$. Thus $g^{xx}=m/\tilde{m}=\mathrm{const}$,  and by rescaling the $x$-coordinate we can always set $g^{xx}=1$. Equation (\ref{compareH1}) therefore reduces to
\begin{align}\label{compareH2}
\frac{1}{2m}g^{IJ}(x,\phi)p_{\phi^I} p_{\phi^J} + V_0(x,\phi)= V_\mathrm{cl}(x,\lambda(\phi,p_\phi)).
\end{align}
It will be useful to analyse this equation at a particular value of $x$, say $x=x_0\in (x_{-},x_+)$, and invert the function $f(\cdot)\equiv V_\mathrm{cl}(x_0,\cdot)$ on the right-hand side. To this end, we will prove the following statement:


\begin{Claim}
For $d\ge 2$, a function $f: \Lambda \to \mathbb{R}$ that satisfies
\begin{align}\label{feq}
f(\lambda(\phi,p_\phi))= \frac{1}{2m}g^{IJ}(x_0,\phi)p_{\phi^I} p_{\phi^J} + V_0(x_0,\phi),
\end{align}
will be invertible.
\end{Claim}
\begin{proof}
Begin by differentiating (\ref{feq}) with respect to $p_{\phi^I}$:
\begin{align}
\frac{\ud f}{\ud \lambda} \frac{\partial \lambda}{\partial p_{\phi^I}}  = \frac{1}{m}g^{IJ}(x_0,\phi)p_{\phi^J}.
\end{align}
Thus $\ud f/\ud \lambda =0\ \Rightarrow\ p_\phi =0$, or equivalently
\begin{align}
\frac{\ud f}{\ud \lambda}&\ne 0& \forall\ \lambda\in \Lambda_{p_\phi\ne0}&\equiv\lambda\left(\left[T^*\mathcal{U}_\phi\right]_{p_\phi\ne 0}\right).
\end{align}
For $d\ge 2$, the set $\left[T^*\mathcal{U}_\phi\right]_{p_\phi\ne 0}$ is connected, and given that $\lambda$ is continuous, $ \Lambda_{p_\phi\ne0}$ will be connected also. Thus, by continuity of $\ud f/ \ud \lambda$ we must have
\begin{align}\label{monostrict}
\text{either}\quad \frac{\ud f}{\ud \lambda}&> 0\quad \forall\ \lambda\in \Lambda_{p_\phi\ne0},\\\label{monostrict2}
\text{or}\quad \frac{\ud f}{\ud \lambda}&< 0\quad \forall\ \lambda\in \Lambda_{p_\phi\ne0}.
\end{align} 
We will treat the first case only -- the second possibility can be dealt with in a similar fashion.

To proceed, we shall construct paths in $[T^*\mathcal{U}_\phi]_{p_\phi\ne 0}$ that explore all $\lambda(\phi,p_\phi)\in\Lambda$ (except perhaps the upper and lower bounds) and then appeal to monotonicity (\ref{monostrict}) to prove that $f$ is invertible. According to the definition (\ref{lambdadef}) for each $\lambda^\star\in \Lambda$ there is at least one $(\phi^\star,p^\star_\phi)\in T^*\mathcal{U}_\phi$ such that $\lambda(\phi^\star,p^\star_\phi)=\lambda^\star$.  So for every $\lambda^\star\in \Lambda$, we can always define a path
\begin{align}\label{path}
\Gamma : \mathbb{R} &\to \left[  T^{*}\mathcal{U}_{\phi }\right]  _{p
_{\phi }\ne 0},& \Gamma (s)&\equiv (\phi ^{\star },p^{\star }_{\phi }+
e^{s} v),
\end{align}
where we choose $v=(v_1,\ldots,v_d)\ne 0$ such that $v\cdot p^\star_\phi \ge 0$. Observe that the path (\ref{path}) has the following properties:
\begin{align}\label{pathimage}
\lambda(\Gamma(\mathbb{R}))&\subseteq \Lambda_{p_\phi\ne 0},
\\
\label{limlow}
\lim_{s\to -\infty} \lambda(\Gamma(s))&= \lambda^\star.
\end{align} 
Furthermore, equation (\ref{feq}) implies
\begin{align}\nonumber
\lim_{s\to\infty}f(\lambda(\Gamma(s))) &= \lim_{s\to\infty}\left\{  \frac{e^{2s}}{2m}g^{IJ}(x_0,\phi^\star)v_I v_J  + O(e^s)\right\}\\ \label{inflimit}
&=\infty.
\end{align}
But recall from (\ref{monostrict}) that $f$ is strictly increasing over $\Lambda_{p_\phi\ne0}$; hence (\ref{inflimit}) requires 
\begin{align}\label{sup1}
\lim_{s\to \infty}\lambda(\Gamma(s))= \sup \{\Lambda_{p_\phi\ne0}\}  \in \mathbb{R}\cup \{\infty\}.
\end{align}
Moreover, because $\Lambda$ is the closure of $\Lambda_{p_\phi\ne0}$ ($T^*\mathcal{U}_\phi$ is the closure of $[T^*\mathcal{U}_\phi]_{p_\phi\ne 0}$, and $\lambda$ is continuous) we have $\sup \{\Lambda_{p_\phi\ne0}\} =\sup \{\Lambda\}$ and hence
\begin{align}\label{limhi}
\lim_{s\to \infty}\lambda(\Gamma(s))=\sup \{\Lambda\}.
\end{align}
We can now assemble our results: recalling (\ref{pathimage}) and noting that the continuous function $\lambda(\Gamma(\cdot))$ has limits (\ref{limlow}) and (\ref{limhi}), we deduce
\begin{align}\label{pathimage2}
\Lambda_{p_\phi\ne 0}\supseteq\lambda(\Gamma(\mathbb{R})) \supseteq (\lambda^\star, \sup\{\Lambda\}).
\end{align}
(Note that the interval on the right is \emph{open}: $\Lambda_{p_\phi\ne 0}$ does not contain $\lambda^\star$ or $\sup\{\Lambda\}$ in general.) Combining (\ref{pathimage2}) with (\ref{monostrict}) we have
\begin{align}
\frac{\ud f}{\ud \lambda} > 0\quad \forall\ \lambda\in (\lambda^\star, \sup \{\Lambda\}),
\end{align}
but as this is true for all $\lambda^\star\in \Lambda$, we conclude that
\begin{align}
\frac{\ud f}{\ud \lambda} > 0\quad \forall\ \lambda\in (\inf\{\Lambda\}, \sup \{\Lambda\}).
\end{align}
Hence $f$ is invertible.
\end{proof}

For $d\ge 2$, we can therefore rewrite (\ref{feq}) as
\begin{align}\label{finverted}
\!\lambda(\phi,p_\phi)= f^{-1}\!\!\left(\!\frac{1}{2m}g^{IJ}(x_0,\phi)p_{\phi^I} p_{\phi^J} + V_0(x_0,\phi)\!\right)\!,\!
\end{align}
and substitute this into equation (\ref{compareH2}):
\begin{align}\nonumber
&\frac{1}{2m}g^{IJ}(x,\phi)p_{\phi^I} p_{\phi^J} + V_0(x,\phi)\\\nonumber
&= V_\mathrm{cl}\left(x,f^{-1}\!\left(\frac{1}{2m}g^{IJ}(x_0,\phi)p_{\phi^I} p_{\phi^J} + V_0(x_0,\phi)\right) \right)\\ \label{Vnew}
&\equiv \tilde{V}_\mathrm{cl}\left(x,\frac{1}{2m}g^{IJ}(x_0,\phi)p_{\phi^I} p_{\phi^J} + V_0(x_0,\phi)\right).
\end{align}
But the only way this equation can hold for all $p_\phi$ is if the function on the right is a first-order polynomial in its second argument: $\tilde{V}_\mathrm{cl}(x,y)=  \alpha(x) + \beta(x)\times y$ for some $\alpha, \beta : (x_-,x_+)\to \mathbb{R}$. Thus (\ref{Vnew}) becomes
\begin{align}\label{linearV}
&\frac{1}{2m}g^{IJ}(x,\phi)p_{\phi^I} p_{\phi^J} + V_0(x,\phi)\\ \nonumber
&= \alpha(x) + \beta(x)\left(\frac{1}{2m}g^{IJ}(x_0,\phi)p_{\phi^I} p_{\phi^J} + V_0(x_0,\phi)\right),
\end{align}
and comparing the coefficients of $p_{\phi^I} p_{\phi^J}$, we see that
\begin{align}\label{gform}
g^{IJ}(x,\phi)&=\beta(x) g^{IJ}(x_0,\phi).
\end{align}
Of course, the metric on the left cannot depend on our choice of $x_0$, so neither can the combination on the right.\footnote{The function $\beta(x)$  implicitly depends on the choice of $x_0$ because $f(\cdot)\equiv V_\mathrm{cl}(x_0,\cdot)$ was used to construct $\tilde{V}_\mathrm{cl}$. This implicit dependence must cancel the explicit dependence of $g^{IJ}(x_0,\phi)$ in (\ref{gform}).} Hence we can write (\ref{gform}) as
\begin{align}
g^{IJ}(x,\phi)&= [b(x)]^{-2}\tilde{g}^{IJ}(\phi),
\end{align}
and conclude that the metric has the form (\ref{genmetric}) as claimed.

We can also compare the terms in (\ref{linearV}) that are independent of $p_\phi$:
\begin{align}\nonumber
V_0(x,\phi)&= \alpha(x) + \beta(x)V_0(x_0,\phi)\\
&\equiv V_x(x) + [b(x)]^{-2}V_\phi(\phi),
\end{align}
revealing that the potential can be slightly more general than the $V_0(x)$ considered in the main text. It is easy to introduce the new term $[b(x)]^{-2}V_\phi(\phi)$ into the analysis of section \ref{DiscGen}: one simply adds $V_\phi$ to the definition of $E_\phi$ at the classical level (\ref{clasE}) and the quantum level (\ref{eigenD}). This has no effect on the quantum correction (\ref{VofS}).

All that remains is to deal with the special case $d=1$. As $[T^*\mathcal{U}_\phi]_{p_\phi\ne 0}$ is disconnected, $\Lambda_{p_\phi >0}\equiv\lambda([T^*\mathcal{U}_\phi]_{p_\phi> 0})$ may not connect to $\Lambda_{p_\phi <0}\equiv\lambda([T^*\mathcal{U}_\phi]_{p_\phi< 0})$ and we might have
\begin{align}
\frac{\ud f}{\ud\lambda} &>0\quad \forall\  \lambda \in \Lambda_{p_\phi >0},
\\
\frac{\ud f}{\ud \lambda} &<0\quad \forall\  \lambda \in \Lambda_{p_\phi <0},
\end{align}
or vice versa. In this case, $f$ will not be invertible, but we can still define partial inverses $f_+^{-1}: f(\Lambda_{p_\phi\ge 0})\to \Lambda_{p_\phi\ge 0}$ and $f_-^{-1}: f(\Lambda_{p_\phi\le 0})\to \Lambda_{p_\phi \le0}$ using the same arguments as $d\ge 2$. We would then write 
\begin{align}
\lambda(\phi,p_\phi) &= f^{-1}_{\mathrm{sign}(p_\phi)}(f(\lambda(\phi,p_\phi)))\\ \nonumber
&=f^{-1}_{\mathrm{sign}(p_\phi)}\!\left(\frac{1}{2m}g^{IJ}(x_0,\phi)p_{\phi^I} p_{\phi^J} + V_0(x_0,\phi)\right),
\end{align}
instead of (\ref{finverted}) and use this in (\ref{Vnew}). Thus, $\tilde{V}_\mathrm{cl}$ becomes $\tilde{V}^{\mathrm{sign}(p_\phi)}_\mathrm{cl}$, but seeing as the left-hand side of equation (\ref{Vnew}) is invariant under $p_\phi \to - p_\phi$, we must have  $\tilde{V}^{+}_\mathrm{cl}=\tilde{V}^{-}_\mathrm{cl}$ anyway. From then on, the argument proceeds exactly as above.

\bibliography{Tube1}
\end{document}